\documentclass[sigconf, nonacm]{acmart}

\usepackage[utf8]{inputenc}
\usepackage{amsmath,amsfonts}

\usepackage{algorithm}
\usepackage{algpseudocode}
\usepackage{enumitem}
\usepackage{subcaption}
\usepackage{graphicx}
\usepackage{booktabs}
\usepackage{textcomp}
\usepackage{geometry}
\usepackage{tikz}
\usetikzlibrary{positioning, shapes, arrows.meta, backgrounds, fit}

\definecolor{sperblue}{RGB}{70, 130, 180}
\definecolor{spergreen}{RGB}{50, 205, 50}
\definecolor{sperred}{RGB}{220, 20, 60}
\definecolor{spergray}{RGB}{200, 200, 200}

\usepackage{tikz}
\usetikzlibrary{shapes, arrows.meta, positioning, calc, backgrounds, fit}
\usepackage{multirow}

\newtheorem{Problem}{problem}

\begin{document}
\title{SPER: Accelerating Progressive Entity Resolution via Stochastic Bipartite Maximization}

%%
%% The "author" command and its associated commands are used to define the authors and their affiliations.
\author{Dimitrios Karapiperis}
\orcid{0000-0002-3878-5988}
\affiliation{%
  \institution{International Hellenic University}
  \city{Thessaloniki}
  \country{Greece}
}
\email{dkarapiperis@ihu.edu.gr}

\author{George Papadakis}
\orcid{0000-0002-7298-9431}
\affiliation{%
  \institution{National and Kapodistrian University of Athens}
  \city{Athens}
  \country{Greece}
}
\email{gpapadis@di.uoa.gr}

\author{Vassilios S. Verykios}
\orcid{0000-0002-9758-0819}
\affiliation{%
  \institution{Hellenic Open University}
  \city{Patras}
  \country{Greece}
}
\email{verykios@eap.gr}

\begin{abstract}
Entity Resolution (ER) is a critical data cleaning task for identifying records that refer to the same real-world entity. In the era of Big Data, traditional batch ER is often infeasible due to volume and velocity constraints, necessitating Progressive ER methods that maximize recall within a limited computational budget. However, existing progressive approaches fail to scale to high-velocity streams because they rely on deterministic sorting to prioritize candidate pairs, a process that incurs prohibitive super-linear complexity and heavy initialization costs. To address this scalability wall, we introduce SPER (Stochastic Progressive ER), a novel framework that redefines prioritization as a sampling problem rather than a ranking problem. By replacing global sorting with a continuous stochastic bipartite maximization strategy, SPER acts as a probabilistic high-pass filter that selects high-utility pairs in strictly linear time. Extensive experiments on eight real-world datasets demonstrate that SPER achieves significant speedups ($3\times$ to $6\times$) over state-of-the-art baselines while maintaining comparable recall and precision.
\end{abstract}

\maketitle

%\pagestyle{\vldbpagestyle}
%\begingroup\small\noindent\raggedright\textbf{PVLDB Reference Format:}\\
%\vldbauthors. \vldbtitle. PVLDB, \vldbvolume(\vldbissue): \vldbpages, %\vldbyear.\\
%\href{https://doi.org/\vldbdoi}{doi:\vldbdoi}
%\endgroup
%\begingroup
%\renewcommand\thefootnote{}\footnote{\noindent
%This work is licensed under the Creative Commons BY-NC-ND 4.0 International License. Visit \url{https://creativecommons.org/licenses/by-nc-nd/4.0/} to view a copy of this license. For any use beyond those covered by this license, obtain permission by emailing \href{mailto:info@vldb.org}{info@vldb.org}. Copyright is held by the owner/author(s). Publication rights licensed to the VLDB Endowment. \\
%\raggedright Proceedings of the VLDB Endowment, Vol. \vldbvolume, No. %\vldbissue\ %
%ISSN 2150-8097. \\
%\href{https://doi.org/\vldbdoi}{doi:\vldbdoi} \\
%}\addtocounter{footnote}{-1}\endgroup

%\ifdefempty{\vldbavailabilityurl}{}{
%\vspace{.3cm}
%\begingroup\small\noindent\raggedright\textbf{PVLDB Artifact Availability:}\\
%The source code, data, and/or other artifacts have been made available %at \url{\vldbavailabilityurl}.
%\endgroup

\section{Introduction}

Entity Resolution (ER) is the critical data cleaning task of identifying and linking database records that refer to the same real-world entity, such as \textit{IBM} and \textit{International Business Machines} \cite{christ_survey, pap_survey}. Typically, 
%this process relies on 
ER involves a \textit{blocking} or \textit{indexing} phase that efficiently retrieves candidate pairs, followed by a computationally intensive \textit{matching} phase that verifies whether they refer to the same entity \cite{peterbook}. In the era of Big Data, where Volume and Velocity are paramount, the traditional \textit{batch} ER process is increasingly infeasible, as it requires processing the entire dataset before producing any output \cite{christ_survey}. For time-sensitive applications, such as real-time fraud detection, 
%or crisis mapping, 
waiting hours for a complete resolution is not an option.

To address this latency, \textit{Progressive ER} redefines the objective: instead of maximizing total recall at the end of a long batch process, it aims to maximize recall \textit{early}, within a limited time or computational budget \cite{Altowim2014Progressive}. Consider a disaster response scenario where thousands of social media reports stream in every minute. A progressive system prioritizes highly probable matches (e.g., exact geolocation overlaps) to identify major clusters immediately, leaving ambiguous fuzzy matches for later refinement if time permits.

Financial crime detection and high-velocity e-commerce exemplify critical domains where the utility of ER decays rapidly, necessitating a \textit{pay-as-you-go} paradigm \cite{Whang2013PayAsYouGo}. In anti-money laundering systems, detecting illicit activity immediately upon data arrival allows investigators to intervene before funds are moved, whereas traditional batch processes that run overnight are often too late to prevent fraud rings that operate across vast transaction logs. Similarly, large online retailers that continuously ingest product data cannot afford to wait for full batch deduplication. Progressive ER addresses this by prioritizing obvious matches like exact UPC links to make the inventory available for sale right away, while processing ambiguous cases in the background.

While state-of-the-art progressive ER methods successfully prioritize matching pairs, they fundamentally fail to scale to large volumes of data, due to a shared algorithmic flaw: their reliance on \textit{deterministic sorting}. 
The main progressive ER frameworks introduced in~\cite{Papenbrock2015Progressive, firmani, Gazzarri2023Incremental, Maciejewski2025DesignSpace,  ondemand, Simonini2019SchemaAgnostic, Galhotra2021Efficient, gsm} depend on heavy initialization phases to strictly rank entities or blocks. For example, the methods in~\cite{Gazzarri2023Incremental, Simonini2019SchemaAgnostic} suffer from: (1) a massive initialization cost because they must construct a meta-blocking graph~\cite{pap1, pap2} and calculate the duplication likelihood for every node pair, and (2) a sorting bottleneck, as they must sort these likelihoods to find possible matches. This requirement imposes a prohibitively high cost in the sense that these methods typically incur hours of initialization latency on large datasets before emitting a single result. Similarly, the latest work in the field 
%a very recent work presented in~
sorts all candidate pairs (edges) in a similarity graph to define the processing order \cite{Maciejewski2025DesignSpace}. These approaches scale super-linearly relative to the number of candidate edges/pairs $\mathcal{E}$: $\mathcal{O}(|\mathcal{E}|\log|\mathcal{E}|)$. In a similar vein, \textit{pBlocking}~\cite{Galhotra2021Efficient} involves an iterative feedback loop: it processes a batch of pairs, it pauses to collect matching results, and it subsequently updates block statistics to restructure the hierarchy. This constant re-evaluation and sorting of blocks to identify the cleanest candidates for the next iteration results in a complexity of $O(n \log^2 n)$ per round, where $n$ is the total number of records.

To address the high initialization cost of Progressive ER, we propose \textit{SPER}, a high-velocity framework that integrates semantic embeddings with a continuous stochastic bipartite maximization strategy to prioritize candidate pairs. By fundamentally redefining prioritization as a weighted sampling problem rather than a global ranking one, our approach shifts the paradigm from deterministic ordering to probabilistic filtering. This enables the immediate, linear-time identification of high-confidence matches with utility statistically equivalent to the optimal ranking, yet without any initialization latency. Even though SPER might not pick the exact same list of pairs as global sorting, there is no significant impact on effectiveness, as the total similarity weight of the selected pairs is equally high. We experimentally demonstrate the insignificant (if any) impact on effectiveness along with the significant gains in time efficiency through an extensive experimental study that involves 8 datasets commonly used in the literature.

More specifically, our approach conveys the following key contributions:
\begin{itemize}[leftmargin=*]
    \item We introduce a novel paradigm for Progressive ER, which replaces the deterministic sorting during the initialization phase with a stochastic process that efficiently retrieves the top-weighted candidate pairs by dynamically scaling selection probabilities to adhere to a strict budget constraint.
    \item We theoretically prove that our stochastic relaxation efficiently retrieves high-weight candidates in \textit{linear} time, bypassing the expensive sorting bottleneck while maintaining theoretical utility guarantees. As a result, it significantly reduces the prioritization complexity compared to the super-linear costs of existing progressive ER methods.
    \item We perform an extensive experimental evaluation that involves eight benchmark datasets. The experimental results demonstrate that SPER eliminates initialization latency entirely, while delivering major gains in time efficiency, reducing the overall run-times from 4$\times$ to $>$6$\times$ across diverse data scales and schema heterogeneities when compared to the state-of-the-art in the field. This is achieved without sacrificing effectiveness, as the cumulative recall and precision are comparable (and at times higher) than the state-of-the-art in the field.
\end{itemize}
The remainder of this paper is organized as follows. Section \ref{sec:related} reviews the major existing work in Progressive ER. Section \ref{sec:problem} formally defines the problem of scalable utility maximization on bipartite graphs. Section \ref{sec:sper} introduces the SPER framework, detailing the stochastic bipartite maximization strategy and providing the theoretical proof of its convergence to the expected utility of the optimal baseline. Section \ref{sec:evaluation} presents the comprehensive experimental evaluation, comparing SPER against state-of-the-art baselines on eight real-world datasets. Finally, Section \ref{sec:conclusions} summarizes our key findings and concludes the work.%Finally, Section \ref{sec:conclusions} concludes the paper and outlines directions for future research.

\section{Related Work}
\label{sec:related}
While ER encompasses various specialized settings described in these surveys~\cite{christ_survey, pap_survey}, research has increasingly prioritized progressive~\cite{Whang2013PayAsYouGo, Altowim2014Progressive, Papenbrock2015Progressive, firmani, Simonini2019SchemaAgnostic, Gazzarri2023Incremental, Galhotra2021Efficient, ondemand, kar-bigdata2021, gsm} and online~\cite{ioannou1,altwaijry, gazzari1, araujo1, gruenheid1,kar-bigdata2020, kar-pvldb, tkde2024} architectures to address the demands of time-sensitive applications.

The concept of \textit{pay-as-you-go} ER was pioneered by Whang et al.~\cite{Whang2013PayAsYouGo}, who proposed maximizing the \textit{early quantity} of detected matches when the computational resources are insufficient to process the entire dataset. To prioritize the pairwise comparisons that are most likely to involve duplicate entities, the proposed solutions leverage heuristics called \textit{hints}, such as sorted lists of record pairs or hierarchies of partitions. Building on this, Altowim et al.~\cite{Altowim2014Progressive} extended progressive techniques to \textit{Relational} ER by dynamically generating resolution plans that rely on cost-benefit models to prioritize decisions with the highest propagation impact. Papenbrock et al.~\cite{Papenbrock2015Progressive} %formalized the core requirements for progressive ER: \textit{improved early quality} and \textit{same eventual quality}, 
introduced dynamic algorithms that iteratively increase sorting windows or process hierarchical blocks based on the latest detected matches. 
All these approaches 
%However, these foundational frameworks predominantly rely on static 
involve sorting heuristics or heavy hierarchical structures that incur high initialization costs of superlinear time complexity. Hence, they struggle to scale to voluminous, high-velocity data, as the overhead of their deterministic ranking of candidate pairs creates a severe bottleneck.

To address the inapplicability of schema-based blocking in heterogeneous big data,
%integration scenarios, 
Simonini et al.~\cite{Simonini2019SchemaAgnostic} introduced a taxonomy of progressive methods and developed algorithms like \textit{PPS} that leverage a blocking graph to prioritize entities without schema knowledge, based exclusively on block co-occurrence patterns. Gagliardelli et al.~\cite{gsm} extended this framework by replacing the heuristic weights with probabilistic classification scores.
Both approaches, though, suffer from high initialization latency: they involve a time-consuming pre-processing phase that  constructs the full blocking graph and then performs global sorting operations to rank entities (or blocks) before emitting the top-weighted pairs. 

To address the trade-off between aggressive and permissive blocking, Galhotra et al.~\cite{Galhotra2021Efficient} proposed \textit{pBlocking}, which implements feedback-driven methodology. Unlike static strategies, which produce a processing order that is independent of detected matches, pBlocking creates a loop where partial ER results refine the processing order in real-time. However, this iterative refinement introduces a \textit{stop-and-wait} bottleneck, as the system must pause to re-rank block collections based on updated scores after every feedback loop, preventing true streaming throughput.

On another line of research, Sun et al.~\cite{Sun2022CostBenefit} proposed \textit{EPEM}, which handles datasets that exceed the capacity of the main memory by using a cost-benefit model to schedule data partitions between disk and memory. On the downside, this approach incurs a significant pre-processing overhead, as it relies on a coarse clustering phase that requires sorting all records based on cumulative similarity. This yields a super-linear cost, while the reliance on disk I/O and the NP-Complete complexity of its partition scheduling logic introduce latency bottlenecks that prevent true real-time processing.

\textit{BrewER}~\cite{ondemand} introduces a query-driven progressive ER framework that prioritizes the resolution of entities that satisfy specific SQL queries (e.g., ORDER BY). While effective for top-$k$ retrieval, it inherently depends on maintaining a global priority queue to enforce a deterministic emission order. This imposes a heap management overhead and \textit{head-of-line} blocking, as the top entity must be fully resolved before emission. These constraints create latency bottlenecks that limit its scalability for general-purpose, high-velocity settings.

Addressing data velocity, \textit{PIER} \cite{Gazzarri2023Incremental} prioritizes comparisons not just within the current data increment but globally across buffered profiles in order to spot duplicates arriving at different times. While this ensures \textit{globality}, it maintains and constantly updates complex global priority queues, which introduce a significant computational bottleneck as the buffer grows. 

Finally, Maciejewski et al.~\cite{Maciejewski2025DesignSpace} systematized the field with a comprehensive \textit{design space exploration}, proposing a unified framework of filtering, weighting, scheduling, and matching. Despite evaluating novel combinations like pre-trained language models, their exploration reaffirmed that scheduling strategies remain bound by the super-linear complexity of deterministic sorting, thus identifying a scalability wall, where memory-intensive join workflows fail to process large datasets.

%Overall, all existing methods suffer from high computational cost in the initialization phase, which significantly restricts their scalability. Our work goes beyond these methods by substituting the exact sorting of candidate pairs with a stochastic prioritization strategy that effectively emulates the optimal order in linear time, thereby eliminating the initialization bottleneck and ensuring scalability

\section{Problem Formulation}
\label{sec:problem}

Let $R$ and $S$ be two distinct collections of entity profiles. We model the resolution space as a Bipartite Similarity Graph $G = (R \cup S, \mathcal{E}, \mathcal{W})$, where:
\begin{itemize}[leftmargin=*]
\item $R$ and $S$ are disjoint sets of vertices ($V = R \cup S$).
\item $\mathcal{E} \subseteq R \times S$ is the set of edges (i.e., candidate pairs) identified by the blocking step.
\item $\mathcal{W}$ is the set of edge weights, where each $w_{(r,s)} \in [0, 1]$ indicates the matching likelihood between profiles $r \in R$ and $s \in S$.

\end{itemize}

In this context, we formulate the task we examine %core optimization challenge 
as follows:
\begin{Problem}[Scalable Utility Maximization]
\label{problem1}
Given a budget $B$ and a computational constraint of linear time complexity $\mathcal{O}(|\mathcal{E}|)$, find a subset of pairs $\mathcal{S}^* \subseteq \mathcal{E}$ with cardinality $|\mathcal{S}^*| \le B$ that maximizes the sum of weights: %{\color{red}such that each vertex appears only once in $\mathcal{E}$}:
\begin{equation}\mathcal{S}^* = \underset{\mathcal{S} \subseteq \mathcal{E}, |\mathcal{S}| \le B}{\arg\max} \sum_{(r,s) \in \mathcal{S}} w_{(r,s)}\end{equation}
\end{Problem}

Note that this definition is generic enough to cover both Record Linkage, where $R$ and $S$ are individually duplicate-free, but overlapping datasets, and Deduplication, where the input comprises a single dataset $R \cup S$ with duplicates in itself. Note also that this task is independent of matching, yielding a static processing order that can be combined with any matching algorithm from the literature.

\section{Approach} 
\label{sec:sper}
To address Problem \ref{problem1}, satisfying its strict linearity constraint, we propose \textit{the Stochastic Progressive Entity Resolution framework} (\textbf{SPER}), which inherently overcomes the scalability limitations of deterministic sorting in high-velocity ER tasks by relaxing the deterministic requirement of finding the \textit{exact} top-$B$ pairs. Unlike the existing progressive methods that typically rely on sorting and ranking, SPER operates as a continuous, probabilistic filter. It applies a \emph{Stochastic Bipartite Maximization} strategy that targets a stochastic relaxation of $\mathcal{S}^*$ by treating edge selection as a sequence of independent Bernoulli trials. By assigning selection probabilities proportional to similarity weights, this approach replaces the global sorting operator with a local sampling filter,
reducing the time complexity of the initialization phase of Progressive ER to $\mathcal{O}(|\mathcal{E}|)$, while concentrating the expected utility on high-weight candidates. As a result, high-value matches are statistically more likely to be processed early in the stream, satisfying the core requirement of Progressive~ER.

To this end, SPER embeds the records of $R$ into dense vectors using an embedding model and then, it stores these embeddings in an Approximate Nearest Neighbor Search (ANNS) index capable of returning the top-$k$ neighbors in logarithmic query time. More specifically, 
%The system architecture consists of 
SPER involves three phases:

\begin{enumerate}[leftmargin=*]
    \item \textbf{Retrieval:} For each entity from $S$, SPER embeds it into a dense vector that is then posed as a query to the ANNS index, retrieving a set of candidate matches from $R$. This generates a local, unranked bipartite subgraph for each new entity (i.e., query). 
    
    \item \textbf{Stochastic Prioritization:} Instead of buffering and sorting these candidates to find the best match, a process of super-linear complexity, SPER applies a Stochastic Bipartite Maximization strategy. It evaluates each candidate pair independently, assigning a selection probability proportional to its similarity weight, defined as the inner product of their L2-normalized embeddings. A lightweight Bernoulli trial (coin flip) determines if the pair is retained or discarded. Each selected pair is added to the set $\mathcal{S}'$ and is subsequently evaluated by a bi-encoder matching function~\cite{sudowoodoo, ember, lsblock, zeakis2023}.
    
    \item \textbf{Budget-Aware Execution:} To respect the global computational budget $B$ without centralized coordination, the system employs a dynamic scaling factor $\alpha$. This factor modulates the selection probabilities, ensuring that the aggregate number of sampled candidate pairs converges to the target budget in expectation, regardless of the number of entities in $S$ or the similarity distribution. 
\end{enumerate}

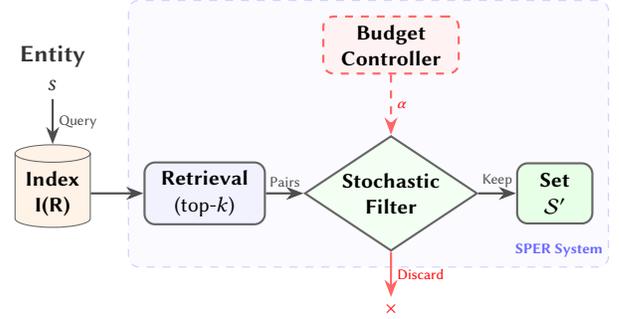
\begin{figure}[t]
\centering
\begin{tikzpicture}[
    % Compact dimensions for single column
    node distance=0.6cm and 0.7cm, 
    >=Stealth,
    font=\sffamily\small,
    % Styles
    process/.style={rectangle, draw=black!60, fill=blue!5, thick, minimum height=0.8cm, minimum width=1.6cm, align=center, rounded corners},
    database/.style={cylinder, draw=black!60, shape border rotate=90, fill=orange!10, aspect=0.25, minimum height=1.0cm, minimum width=1.0cm, align=center},
    decision/.style={diamond, draw=black!60, fill=green!5, thick, aspect=1.5, align=center, inner sep=1pt, minimum width=1.5cm},
    data/.style={rectangle, draw=none, fill=none, align=center, text=black!80, font=\sffamily\bfseries},
    control/.style={rectangle, draw=red!60, fill=red!5, dashed, thick, align=center, rounded corners, minimum height=0.7cm, minimum width=1.8cm},
    arrow/.style={->, thick, color=black!70},
    line/.style={-, thick, color=black!70}
]

% --- Nodes ---

% 1. Index (Bottom Left)
\node[database] (index) {\textbf{Index}\\\textbf{I(R)}};

% 2. Input Stream (Moved ABOVE Index to save width)
\node[data, above=0.6cm of index] (stream) {Entity\\$s$};

% 3. Retrieval Process (Right of Index)
\node[process, right=of index] (retrieval) {\textbf{Retrieval}\\(top-$k$)};

% 4. The Stochastic Filter (Right of Retrieval)
\node[decision, right=0.5cm of retrieval] (filter) {\textbf{Stochastic}\\\textbf{Filter}};

% 5. Budget Controller (Above Filter - balances the height with Stream node)
\node[control, above=0.8cm of filter] (budget) {\textbf{Budget}\\\textbf{Controller}};

% 6. Output (Right of Filter)
\node[process, right=0.5cm of filter, fill=green!10, minimum width=1.0cm] (output) {\textbf{Set}\\$\mathcal{S}'$};

% --- Connections ---

% Stream to Index (Vertical down)
\draw[arrow] (stream) -- node[right, scale=0.7] {Query} (index);

% Index to Retrieval (Horizontal)
\draw[arrow] (index) -- node[above, scale=0.6] {} (retrieval);

% Retrieval to Filter
\draw[arrow] (retrieval) -- node[above, scale=0.7, align=center] {Pairs} (filter);

% Filter to Output (YES)
\draw[arrow] (filter) -- node[above, scale=0.7] {Keep} (output);

% Filter Discard (NO) - Downwards
\coordinate[below=0.6cm of filter] (trash);
\draw[arrow, color=red!60] (filter) -- node[right, scale=0.7, color=red] {Discard} (trash) node[below, scale=0.8, color=red] {$\times$};

% Control Feedback (Curve or straight line)
\draw[arrow, color=red!60, dashed] (budget) -- node[right, scale=0.7, color=red] {$\alpha$} (filter);

% --- Grouping (The SPER System Box) ---
\begin{scope}[on background layer]
    \node[draw=blue!20, fill=blue!2, dashed, 
    fit=(retrieval) (filter) (budget) (output), 
    inner sep=0.2cm, 
    rounded corners,
    label={[anchor=south east, color=blue!60, scale=0.7, yshift=0.1cm]south east:\textbf{SPER System}}] (box) {};
\end{scope}

\end{tikzpicture}
\caption{The \textbf{SPER} framework. Entities from $S$ query the Index, generating candidates for the Stochastic Filter, which dynamically accepts/rejects pairs based on the budget controller.}
\label{fig:sper_architecture}
\end{figure}

Figure \ref{fig:sper_architecture} illustrates the high-level architecture of the SPER framework. All entities from $S$, embedded into dense vectors, are matched against the index $I(R)$, which stores the embeddings of $R$ and returns the top-$k$ results per query, yielding a total of $n=k\cdot|S|$ candidate pairs. The core innovation is the \textit{Stochastic Filter}, represented by the decision diamond, which replaces the traditional blocking priority queue found in deterministic progressive approaches. Instead of buffering and ranking the candidate pairs with $\mathcal{O}(n \log n)$ complexity, the filter makes an instantaneous $\mathcal{O}(1)$ decision for each pair. This design ensures that the system maintains a consistent verification throughput aligned with budget $B$, effectively decoupling the processing latency from the volume of input data.

A natural alternative to stochastic selection is a deterministic policy that retains candidate pairs exceeding a fixed similarity threshold (e.g., $0.8$). While simple, this approach is suboptimal for high-velocity Progressive ER for three reasons:

\begin{enumerate}[leftmargin=*]
    \item \textbf{Latency Indeterminacy:} Strictly selecting the top-$B$ pairs requires observing the entire candidate set to establish a ranking, forcing a batch-processing model with $\mathcal{O}(n \log n)$ sorting costs, where $n$ is the total number of candidate pairs (i.e., $n=k\cdot|S|$). Stochastic sampling approximates the utility of this optimal selection in $\mathcal{O}(n)$ time.
    
    \item \textbf{Budget Rigidity:} A static threshold cannot adapt to data variance. In high-similarity candidates, it may select excessive pairs (violating budget $B$), while in low-similarity candidates, it may starve the verification process. 
    
    \item \textbf{Recall of Ambiguous Matches:} Deterministic thresholds impose a hard cutoff, permanently discarding valid matches that fall slightly below them due to noise (e.g., typos). Stochastic selection maintains a non-zero probability $P$ of selecting lower-confidence pairs, enabling the recovery of subtle duplicates that rigid filtering would miss.
\end{enumerate}

\subsection{Stochastic Bipartite Maximization}
\label{sec:theory}
To approximate the optimal set $\mathcal{S}^*$ without incurring the sorting cost, we propose a \textit{Stochastic Bipartite Maximization} strategy. Crucially, our approach bypasses the construction of any physical graph structure; instead, we treat the selection of each retrieved pair $(r,s)$ as an independent Bernoulli trial. We retain the notation $w_{(r,s)}$ to denote the similarity score (or weight) of a candidate pair, defining a sampling probability $P[X_{(r,s)} = 1]$ that is directly proportional to this weight. Let $\mathcal{S}^*$ be the optimal set of $B$ pairs with the maximum total weight, while the total utility of a set $\mathcal{S}$ is defined as $U(\mathcal{S}) = \sum_{(r,s) \in \mathcal{S}} w_{(r,s)}$. In the following, we prove that our stochastic process generates a random solution set $\mathcal{S}' = \{(r,s) \in \mathcal{E} \mid X_{(r,s)} = 1\}$, which captures high-utility pairs with high probability, approximating the optimal objective function in linear time $\mathcal{O}(|\mathcal{E}|)$.

Treating $B$ as a target average, allows the algorithm to process slightly more or fewer pairs than its specified value per query entity, simplifying the implementation by removing the need to rigorously normalize probabilities to hit an exact count. In a standard Bernoulli process where $P[X_{(r,s)} = 1] = w_{(r,s)}$, the expected number of selected pairs $\mathbb{E}[|S'|] = \sum_{(r,s)} w_{(r,s)}$ depends solely on the data distribution, potentially leading to budget overflow $\mathbb{E}[|S'|] \gg B$.

To address this, we select a candidate pair \((r,s)\) independently with probability: \[ p_{(r,s)} \;=\; \alpha\ \cdot w_{(r,s)}, \]
where the global scaling factor \(0<\alpha \leq 1\) is chosen so that:
\begin{equation}
\label{eq:alpha-constraint}
\sum_{(r,s)} p_{(r,s)} \;=\;  B\implies \alpha = \frac{B}{\sum w_{(r,s)}}.
\end{equation}

We first establish the target budget $B$ as a fixed fraction $\rho$ of the total expected candidate volume, such that $B = \rho \cdot k|S|$. Ideally, the scaling factor $\alpha$ would be set to satisfy the constraint $\sum p_{(r,s)} = B$ exactly; however, computing this optimum is time-consuming for large-scale input datasets.

Instead, SPER initializes $\alpha$ using a conservative estimate derived directly from the budget definition: $\alpha \approx B / (0.5 \cdot k|S|) = 2\rho$, where $0.5$ serves as a safety prior for the average similarity weight. This initialization ensures that our approach begins in a state of controlled under-utilization and ramps up, avoiding an initial overflow that would require drastic correction. Subsequently, to maintain this budget dynamically, SPER employs an online adaptive calibration where candidate pairs are processed in windows of size $W$. After each window, we compare the observed number of selections $m_w$ to a small fixed target $B_w = \left\lceil B \cdot \frac{W}{|S|} \right\rceil$ and update $\alpha$ multiplicatively:

\begin{equation}
\label{eq:alpha}
\alpha_{\text{new}} = \alpha_{\text{old}}  \Bigl(1 + \eta   
\frac{B_{\mathrm{w}} - m_{\mathrm{w}}}{B_{\mathrm{w}}}\Bigr),
\end{equation}
where \(\eta\in(0,1]\) is a small adaptation rate, e.g., $\eta=0.05$. This controller requires no knowledge of the total number of candidate pairs and stabilizes \(\alpha\) quickly. 

While the choice of $\alpha$ establishes the baseline selection probability, the window size $W$ dictates the stability of the control loop. To prevent signal starvation—where a window yields zero selected candidates, causing controller oscillation—the window size must satisfy $W \gg 1 / \rho$.\footnote{The theoretical lower bound $W \ge 1/\rho$ ensures an expected selection count of at least one. However, due to stochastic variance, a tighter practical bound (e.g., $W \ge 5/\rho$) is required to ensure the probability of an empty window remains negligible ($<1\%$).} Beyond stability, extreme values for $\alpha$ can also compromise the controller's precision regarding budget adherence. Specifically, an excessively low $\alpha$ tends to dampen the selection efficiency (potentially dropping to $\approx$$B/2$), while an excessively high $\alpha$ amplifies the variance of the selection process, leading to significant overshoots (e.g., up to $2 \cdot B$).

Having fixed the selection probabilities, let $m$ denote the random number of candidate pairs selected by the algorithm. Modeling the selection process as a sequence of independent Bernoulli trials\footnote{Given that the probabilities vary for each candidate pair, the random variable $m$ follows a Poisson Binomial Distribution.}, each with parameter $p_{(r,s)}$, its expectation and variance are:
\[
\mathbb{E}[m] \;=\; \sum_{(r,s)} p_{(r,s)} \;=\; B,
\qquad
\mathrm{Var}[m] \;=\; \sum_{(r,s)} p_{(r,s)}(1-p_{(r,s)}).
\]
Since \(0\le p_{(r,s)}\le 1\), the following simple upper bound holds:
\[
\mathrm{Var}[m] \le \sum_{(r,s)} p_{(r,s)} = B,
\]
so the standard deviation of $m$ satisfies \(\sigma(m)\le \sqrt{B}\) and the relative standard deviation obeys \(\sigma(m)/B \le 1/\sqrt{B}\).
Thus, for large budgets the random fluctuations are relatively small.

To rigorously quantify this stability, Chernoff bounds for sums of independent Bernoulli variables (with mean \(\mu=\mathbb{E}[m]=B\)) yield, for \(0<\epsilon\le 1\),
\begin{equation}
\Pr\big(|m-B|\ge \epsilon B\big) \le 2\exp\!\left(-\frac{\epsilon^2\cdot B}{3}\right).
\label{eq:chernoff-budget}
\end{equation}
Equation~\eqref{eq:chernoff-budget} implies exponentially small tail probabilities once \(B\) is moderately large (e.g., for \(B=10{,}000\) and \(\epsilon=0.05\), the right-hand side expression is below \(10^{-3}\)).

Given these concentration guarantees, we now analyze the underlying optimization objective implicitly solved by the stochastic selection process. Rather than deterministically solving Problem~\ref{problem1}, SPER optimizes a stochastic relaxation in which the budget constraint is enforced in expectation and candidate pairs are selected probabilistically.

\begin{theorem}[Expected Utility under Stochastic Budgeted Sampling]
\label{thm:utility}
Let $\mathcal{S}'$ be the random set of pairs selected independently with
$P[X_{(r,s)}=1] = \alpha\, w_{(r,s)}$, where $\alpha$ is calibrated so that
$\mathbb{E}[|\mathcal{S}'|]=B$.
Then the expected utility of $\mathcal{S}'$ satisfies:
\begin{equation}
    \mathbb{E}[U(\mathcal{S}')] = \alpha \sum_{(r,s)} w_{(r,s)}^2 .
\end{equation}
This objective favors high-weight pairs and increasingly concentrates utility on top-ranked candidates as the similarity distribution becomes heavy-tailed.\footnote{Since the expected utility scales with the second moment of the weights, high-similarity pairs contribute quadratically more to the objective than low-similarity ones. This non-linear scaling concentrates the selection probability on the rare, high-weight pairs.}
\end{theorem}

\begin{proof}
Let $X_{(r,s)}$ be the indicator variable for selecting pair $(r,s)$. The total utility is $U(\mathcal{S}') = \sum X_{(r,s)} w_{(r,s)}$.
Since SPER sets selection probability $P[X_{(r,s)}=1] = \alpha \cdot w_{(r,s)}$, where $\alpha$ is the scaling factor calibrated to satisfy the budget constraint $\alpha \sum w_{(r,s)} \approx B$, the expected utility is:
\begin{equation}
    \mathbb{E}[U(\mathcal{S}')] = \sum_{(r,s)} \alpha \cdot w_{(r,s)} \cdot w_{(r,s)} = \alpha \sum_{(r,s)} w_{(r,s)}^2
\end{equation}
By the Cauchy-Schwarz bound, $\sum w^2 \geq (\sum w)^2 / (k\cdot |S|)$, implying that emphasizing $w^2$ concentrates utility more strongly on high-weight pairs than uniform sampling, which weights all candidates equally. The algorithm thus optimizes the expected utility under the stochastic budget constraint, acting as a high-pass filter that suppresses low-confidence pairs, while preserving high-similarity candidates.
\end{proof}

Problem~\ref{problem1} requires a deterministic selection of the exact top-$B$ pairs, which in turn necessitates global ranking and sorting.
SPER instead targets a stochastic relaxation of this objective, in which the budget constraint is satisfied in expectation and prioritization is achieved probabilistically. Theorem~\ref{thm:utility} shows that this relaxation maximizes the expected utility in proportion to the second moment of the similarity distribution. As a result, $\mathcal{S}'$ forms a concentrated subset of high-similarity candidates that approximates the top-$B$ solution in practice, while avoiding the super-linear complexity of deterministic scheduling.
This approximation becomes increasingly accurate in ER settings, where true matches typically exhibit a heavy-tailed similarity distribution.

\begin{algorithm}[t]
\caption{Stochastic Bipartite Maximization with Dynamic Budgeting}
\label{alg:stochastic_res}
\begin{algorithmic}[1]
\State \textbf{Input:} Dataset $S$, indexed dataset $R$ as $I$, neighbors $k$, window size $W$, selection percentage $\rho$, embedding model $\mathcal{T}$
\State \textbf{Output:} Selected Pairs $\mathcal{S}'$

\State $\mathcal{S}' \leftarrow \emptyset, m_w \leftarrow 0, \text{count} \leftarrow 0$, $\alpha \leftarrow 2\cdot \rho$  \label{alg:init1}
\State $B \leftarrow \rho \cdot k \cdot |S|$, $B_w = \left\lceil B \cdot \frac{W}{|S|} \right\rceil$, $\eta \leftarrow 0.05$  \label{alg:init2}

\For{each entity $s \in S$}  \label{alg:q1} 
    \State $v \leftarrow \mathcal{T}(s)$ \Comment{Entity $s$ embedded into a dense vector}    
    \State $\mathcal{C}_s \leftarrow I.\text{query}(v, k)$ \Comment{Retrieve top-$k$ candidate Id's of $R$} \label{alg:q2}
    
    \For{each $(r, w) \in \mathcal{C}_s$} \label{alg:sel1}
        \State $P \leftarrow \alpha \cdot  w$ \Comment{Calculate probability of selection by scaling the similarity score}
        \State $u \sim \text{Uniform}(0, 1)$ \label{alg:u1}
        
        \If{$u < P$}    \label{alg:u2}
            \State $\mathcal{S}'.\text{add}((r,s))$ \Comment{Add tuple $(r,s)$ to $\mathcal{S}'$}
            \State $m_w \leftarrow m_w + 1$ \Comment{Track selections}
        \EndIf
    \EndFor \label{alg:sel2}
    
    \State $\text{count} \leftarrow \text{count} + 1$ \label{alg:budget1}
    
    \If{$\text{count} \mod W == 0$} \Comment{End of window}
        \State $\alpha_{\text{new}} \leftarrow \alpha_{\text{old}}   (1 + \eta   \frac{B_w - m_w}{B_w})$ \Comment{Update scaling factor $\alpha$}
        \State $m_w \leftarrow 0$ \Comment{Reset window counter}
    \EndIf \label{alg:budget2}
\EndFor

\State \Return $\mathcal{S}'$
\end{algorithmic}
\end{algorithm}

\subsection{The SPER Algorithm}

Our overall approach is outlined in Algorithm \ref{alg:stochastic_res}, which processes the entities of $S$ in windows of size $W$, while maintaining a running count of selected pairs $m_w$ to adjust $\alpha$ after each window. The algorithm consists of the following steps:

\begin{enumerate}[leftmargin=*]
    \item \textbf{Initialization (Lines \ref{alg:init1} and \ref{alg:init2}):} The process begins by setting the scaling factor to $\alpha = 2 \cdot \rho$. We also instantiate the tracking counters ($count, m_w$), the budget parameters ($B, B_w$) and assign a conservative value of $\eta = 0.05$ to the adaptation rate, ensuring the control loop prioritizes stability and effectively smooths out short-term stochastic variance.
    
    \item \textbf{Retrieval (Lines \ref{alg:q1}--\ref{alg:q2}):} For each query entity $s$ in $S$, SPER embeds it into a dense vector using embedding model $\mathcal{T}$ and queries the index $I$ to retrieve the set $\mathcal{C}_s$ of the Id's of its top-$k$ nearest candidates of $R$.\footnote{Embedding and retrieval are batched operations.} %{\color{red}I would remove Continuous}{\color{blue}OK}
    
    \item \textbf{Stochastic Selection (Lines \ref{alg:sel1}--\ref{alg:sel2}):} For each candidate $r \in \mathcal{C}_s$, the selection probability is calculated as $P = \alpha \cdot  w$. A Bernoulli trial determines if the pair $(r,s)$ is added to the output set~$\mathcal{S}'$. 
    
    \item \textbf{Dynamic Budget Control (Lines \ref{alg:budget1}--\ref{alg:budget2}):} To enforce the budget constraint, the algorithm monitors the selection rate of candidate pairs. After processing a window of $W$ candidate pairs, it adapts $\alpha$ using Equation \ref{eq:alpha}. This feedback loop stabilizes the selection pressure, increasing $\alpha$ if the system is under-budget and decreasing it if being over-budget.
\end{enumerate}

The total runtime is composed of two phases: (1) the retrieval phase, where for every entity in $S$, querying the index takes logarithmic time with respect to the index size $|R|$. Across all queries, this sums to $\mathcal{O}(|S| \cdot \log |R|)$, and (2) the selection phase, where for each of the $|S|$ queries, the algorithm processes $k$ candidates. The stochastic check at Lines \ref{alg:u1} and \ref{alg:u2} is a constant time operation $\mathcal{O}(1)$. The total selection effort is thus $\mathcal{O}(k\cdot|S|)$. Therefore, the combined time complexity is $\mathcal{O}(|S| \cdot \log|R| + k\cdot|S|)$.

Since the total number of candidate edges is $|\mathcal{E}| = k \cdot |S|$, our selection phase scales linearly as $\mathcal{O}(|\mathcal{E}|)$, strictly dominating the super-linear complexity, $\mathcal{O}(|\mathcal{E}| \log |\mathcal{E}|)$, that is required by sorting-based approaches.

\begin{table}[t]
    \centering
    \caption{Characteristics of the 8 benchmark datasets. $|M|$ represents the number of true matching pairs.}
    \label{tab:datasets}
    \vspace{-5pt}
    \resizebox{\columnwidth}{!}{%
    \begin{tabular}{llrrr}
        \toprule
        \textbf{Dataset} & \textbf{Domain} & \textbf{$|S|$} & \textbf{$|R|$} & \textbf{$|M|$} \\
        \midrule
        \textbf{Abt-Buy}\footnotemark[1] & E-Commerce & 1,081 & 1,092 & 1,097 \\
        \textbf{Amazon-Google}\footnotemark[1]  & E-Commerce & 1,363 & 3,226 & 1,300 \\
        \textbf{DBLP-ACM}\footnotemark[1]  & Bibliographic & 2,294 &  2,614& 2,224 \\
        \textbf{DBLP-Scholar}\footnotemark[1] & Bibliographic & 2,616 & 64,263 & 5,347 \\
        \textbf{Walmart-Amazon}\footnotemark[2] & E-Commerce & 2,554 & 22,074 & 1,154 \\
        \textbf{DBPEDIA-IMDB}\footnotemark[3] & Movies & 23,182 & 27,614 & 22,862 \\
        \textbf{NC-Voters}\footnotemark[4] (Semi-synthetic) & Civic & 1M & 1M & 1M \\
        \textbf{DBLP}\footnotemark[5] (Semi-synthetic) &Bibliographic & 3M & 3M & 1.5M \\
        \bottomrule
    \end{tabular}%
    }
\end{table}
\footnotetext[1]{\url{https://old.dbs.uni-leipzig.de/research/projects/object_matching/benchmark_datasets_for_entity_resolution}}
\footnotetext[2]{\url{https://hpi.de/naumann/projects/repeatability/datasets/amazon-walmart-dataset.html}}
\footnotetext[3]{\url{https://zenodo.org/record/8433873/files/data_ea.tar.gz}} 
\footnotetext[4]{\url{https://www.ncsbe.gov/results-data/voter-registration-data}}
\footnotetext[5]{\url{https://dblp.org/xml}}

\begin{figure*}[t]
    \centering
    
    % --- Row 1 ---
    \begin{subfigure}[b]{1.0\textwidth}
        \centering
        \includegraphics[width=\linewidth]{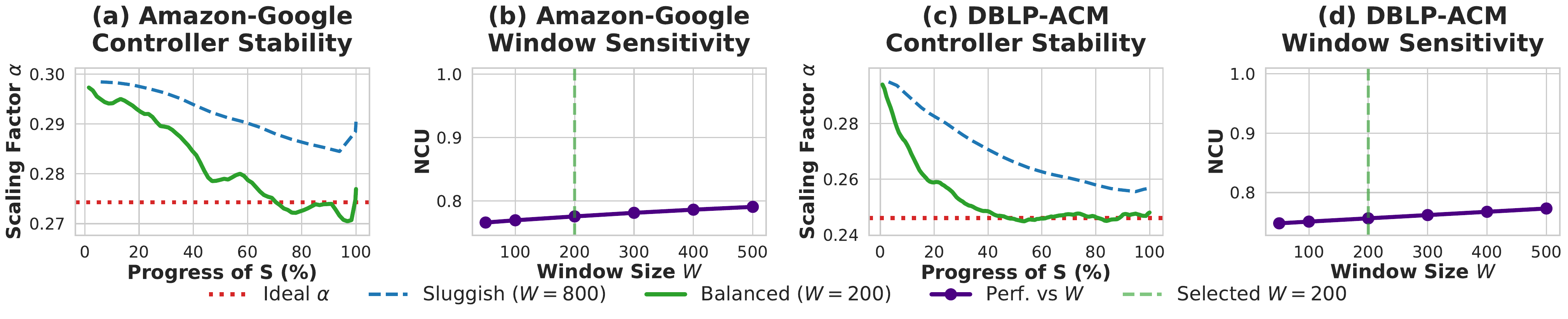}
        %\caption{Description of first plot}
        \label{fig:sens1}
    \end{subfigure}
    \par\bigskip % Forces a new line and adds vertical space
    
    % --- Row 2 ---
    \begin{subfigure}[b]{1.0\textwidth}
        \centering
        \includegraphics[width=\linewidth]{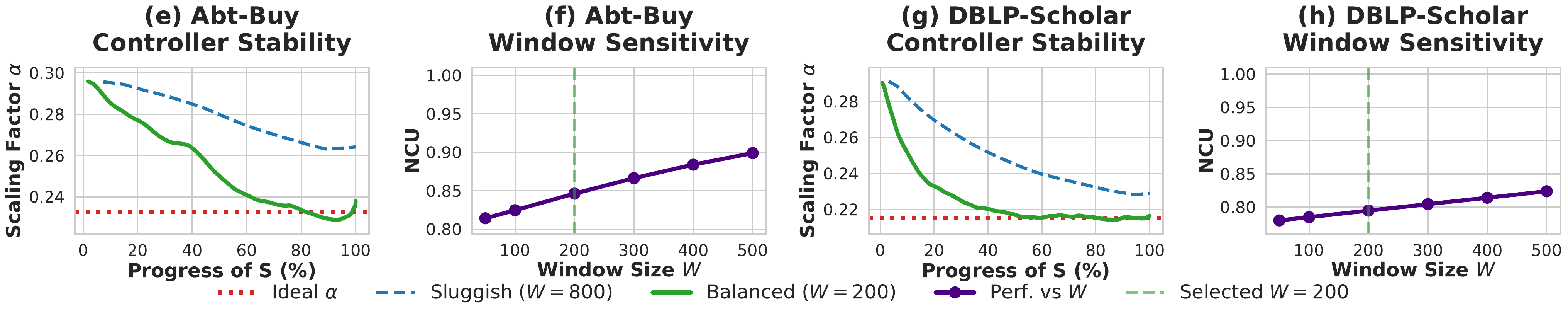}
        %\caption{Description of second plot}
        \label{fig:sens2}
    \end{subfigure}
    \par\bigskip
    
    % --- Row 3 ---
    \begin{subfigure}[b]{1.0\textwidth}
        \centering
        \includegraphics[width=\linewidth]{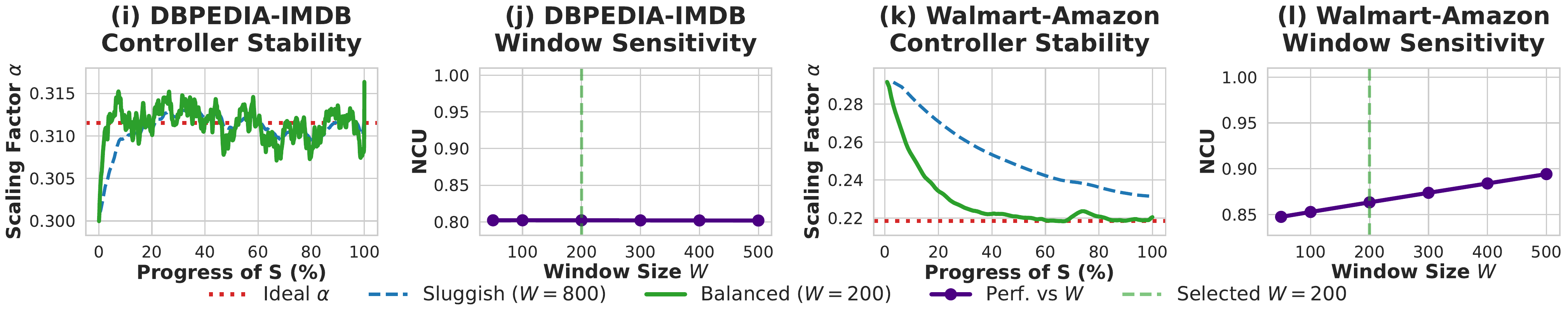}
        %\caption{Description of third plot}
        \label{fig:sens3}
    \end{subfigure}
    \par\bigskip
    
    % --- Row 4 ---
    \begin{subfigure}[b]{1.0\textwidth}
        \centering
        \includegraphics[width=\linewidth]{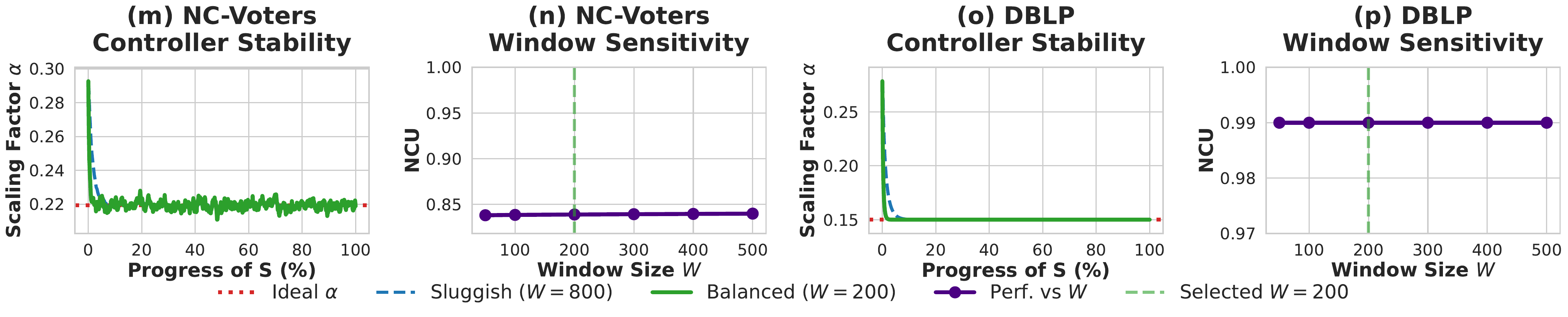}
        %\caption{Description of fourth plot}
        \label{fig:sens4}
    \end{subfigure}
    
    \caption{Parameter Stability and Sensitivity Analysis.}
    \label{fig:sensitivity}
\end{figure*}

Given that Algorithm \ref{alg:stochastic_res} processes $S$ sequentially and makes immediate inclusion decisions, it never materializes the full list of candidates. The working memory requirement is $\mathcal{O}(k)$ (effectively $\mathcal{O}(1)$ relative to $|S|$), whereas the existing, sorting-based Progressive ER methods require $\mathcal{O}(k\cdot|S|)$ space to store and rank all these pairs before further processing occurs.

\section{Experimental Evaluation}
\label{sec:evaluation}
In this section, we provide a comprehensive empirical evaluation of the SPER framework, assessing its performance against the three main state-of-the-art progressive ER techniques. Our experimental analysis is driven by three primary objectives: (1) to validate the scalability of Stochastic Bipartite Maximization, demonstrating its ability to process large datasets in linear time; (2) to verify the stability of the dynamic budget controller, ensuring it strictly adheres to computational constraints across diverse data distributions; (3) to demonstrate the superiority of SPER over the baseline methods with respect to both effectiveness and time efficiency.

To rigorously evaluate the efficiency of the prioritization strategy in isolation—independently of any subsequent matching function—we employ Recall@$B$ (\textit{Cumulative Recall}) and Precision. These metrics measure the proportion of ground-truth matches retrieved and the density of valid pairs, respectively, within the specific budget $B$ of candidates selected by the filter. To further quantify how closely our stochastic selection lies from the theoretical optimum, we report \textit{Normalized Cumulative Utility} (\textbf{NCU}), which is defined as the ratio of the total similarity weight of the selected pairs to that of the ideal top-$B$ subset identified by an offline oracle. Finally, \textit{Execution Time} is recorded to validate the framework's strict linear scalability against super-linear baselines. 

Table~\ref{tab:datasets} summarizes the characteristics of the eight benchmark datasets employed in our evaluation, which are commonly used in the literature~\cite{Simonini2019SchemaAgnostic, Gazzarri2023Incremental, ondemand, kar-bigdata2021, gazzari1, araujo1,kar-bigdata2020, kar-pvldb, tkde2024, gsm}.

To evaluate SPER, we compare it against three state-of-the-art \textit{progressive baselines} that employ distinct prioritization strategies for managing the trade-off between efficiency and result quality: (1) The \textit{I-PES algorithm} (\textbf{PES}) \cite{Gazzarri2023Incremental} operates in an entity-centric that
%addresses dynamic data streams, 
maximizes early quality by adaptively scheduling comparisons based on the match likelihood.
%in an entity-centric way.
%while managing the throughput of {\color{red}incoming} data increments. 
(2) \textit{pBlocking} (\textbf{PBL}) \cite{Galhotra2021Efficient} %operates on static datasets by 
, combined with a perfect matcher, iteratively refines blocking via a feedback loop, using partial ER results to dynamically rescore blocks and prune non-matching pairs. (3) \textit{BrewER} (\textbf{BRW}) \cite{ondemand} implements a query-driven approach, prioritizing the resolution of entities based on specific SQL ORDER BY predicates to support top-$k$ retrieval and early termination without processing the full dataset.

All experiments were conducted on a machine equipped with $80$ GB of RAM and an NVIDIA GPU with $23$ GB of VRAM. We implemented SPER using the Hierarchical Navigable Small Worlds (HNSW) ANNS index~\cite{hnsw}, offered by FAISS\footnote{\url{https://github.com/facebookresearch/faiss}}, which uses highly-optimized operations for both index construction and logarithmic query time. We employ the $384$-dimensional \texttt{MiniLM-L6-v2} embedding model~\cite{mini}, due to its optimal trade-off between inference latency and representation quality~\cite{zeakis2023, lsblock}. We first embed dataset $R$ via a one-time, GPU-accelerated batch operation, followed by the successive embedding of query entities from $S$, where we have set $k=5$ throughout the whole evaluation process. To ensure the reliability of the experiments, the presented results are the average values from $10$ experimental runs.

\subsection{Experimental results}

Figures~\ref{fig:sensitivity}(a), \ref{fig:sensitivity}(c), \ref{fig:sensitivity}(e), \ref{fig:sensitivity}(g), \ref{fig:sensitivity}(i), \ref{fig:sensitivity}(k), \ref{fig:sensitivity}(m), and \ref{fig:sensitivity}(o) compare the trajectory of the scaling factor $\alpha$ against the ideal value of $\alpha$ (red dotted line), which would perfectly calibrate the filter's strictness so that the total number of selected pairs would equal budget $B$ exactly. For these experiments, we fix the target budget ratio at $\rho=0.15$, resulting in an initial scaling factor estimate of $\alpha \approx 0.3$. For Amazon-Google, the \textit{balanced} controller (green solid line, $W=200$) successfully tracks the ideal threshold ($\alpha \approx 0.275$), rapidly correcting initial estimates while maintaining necessary reactivity. In contrast, the \textit{sluggish} baseline ($W=800$, blue dashed line) acts as an excessive low-pass filter. Despite using a reactive adaptation rate ($\eta=0.05$), the large window introduces structural inertia that prevents the controller from adapting to the true density in time. This lag leaves the estimation significantly offset from the optimal operating point during density transitions. A similar pattern is observed on DBLP-ACM, where the balanced controller quickly adjusts to the lower required $\alpha \approx 0.25$, whereas the sluggish baseline fails to converge over the course of processing $S$.

To validate our parameter selection, the sensitivity plots (Figures~\ref{fig:sensitivity}(b), \ref{fig:sensitivity}(d), \ref{fig:sensitivity}(f), \ref{fig:sensitivity}(h), \ref{fig:sensitivity}(j), \ref{fig:sensitivity}(l), \ref{fig:sensitivity}(n), and \ref{fig:sensitivity}(p)) measure the impact of $W$ on the NCU\footnote{The normalization scales this value by dividing it by $U(\mathcal{S}^*)$, the optimal utility for the budget, allowing for a percentage-based comparison (0 to 1.0).}, focusing on the critical operational range $W \in [100, 500]$. We observe that NCU follows a stable high-performance plateau (near $0.8$) across this interval. The results confirm that our choice of $W=200$ (marked by the vertical green line) sits safely within this optimal region, avoiding the noise of smaller windows, while preserving the agility required to track dynamic shifts.

\begin{figure}[t]
    \centering
    % --- Row 1 ---
    \begin{subfigure}[b]{0.48\textwidth}
        \centering
        \includegraphics[width=\linewidth]{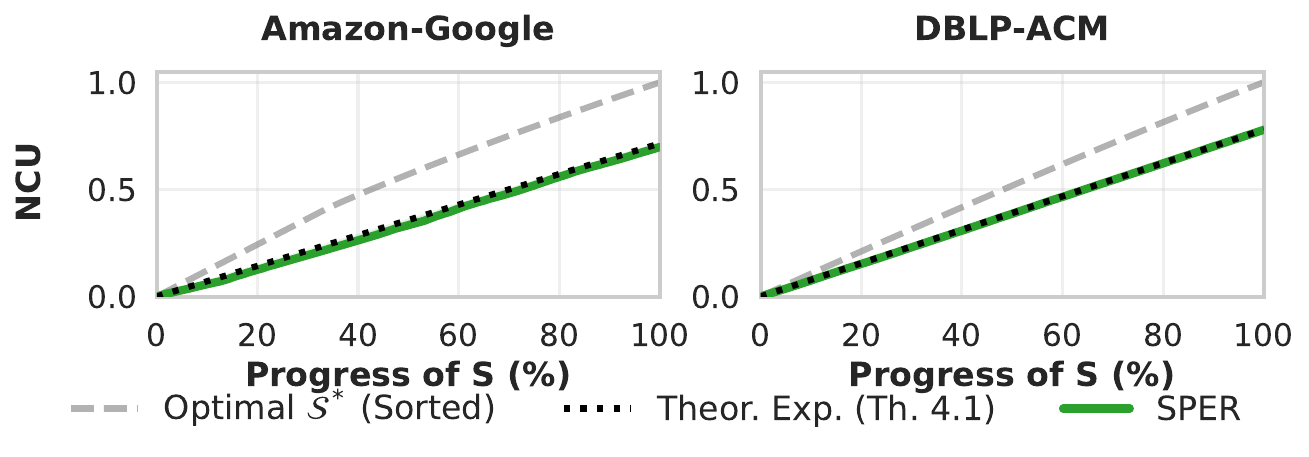}
        %\caption{Description of first validation}
        \label{fig:theo1}
    \end{subfigure}
    \hfill % Pushes the two images apart
    \begin{subfigure}[b]{0.48\textwidth}
        \centering
        \includegraphics[width=\linewidth]{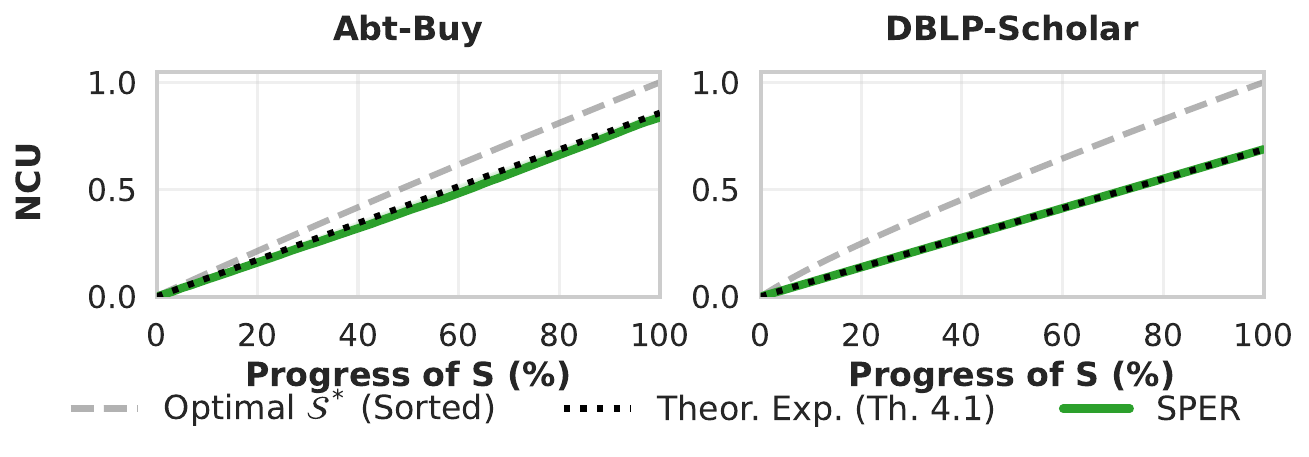}
        %\caption{Description of second validation}
        \label{fig:theo2}
    \end{subfigure}
    
    \par\bigskip % Vertical space between rows
    
    % --- Row 2 ---
    \begin{subfigure}[b]{0.48\textwidth}
        \centering
        \includegraphics[width=\linewidth]{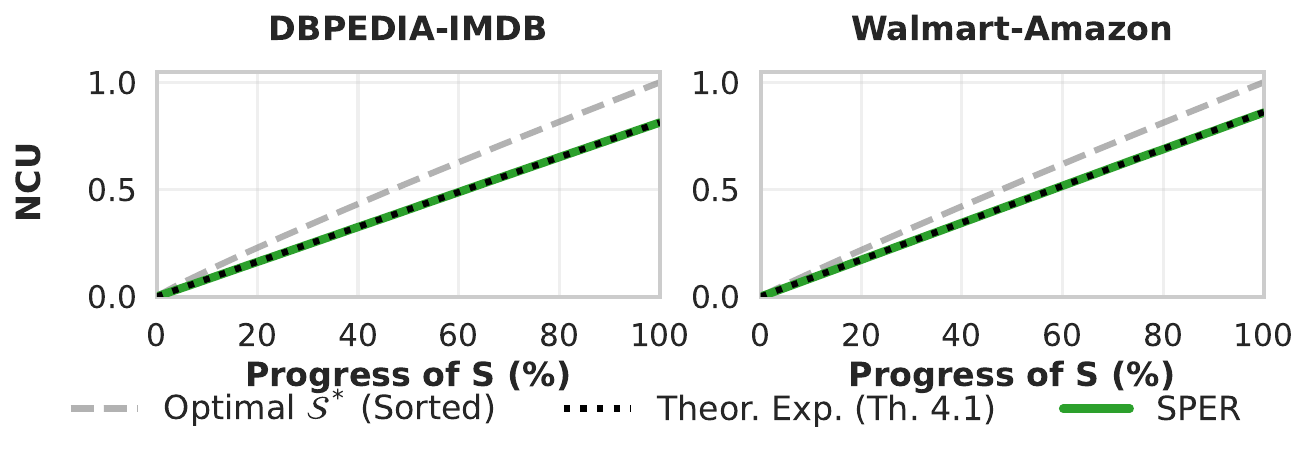}
        %\caption{Description of third validation}
        \label{fig:theo3}
    \end{subfigure}
    \hfill
    \begin{subfigure}[b]{0.48\textwidth}
        \centering
        \includegraphics[width=\linewidth]{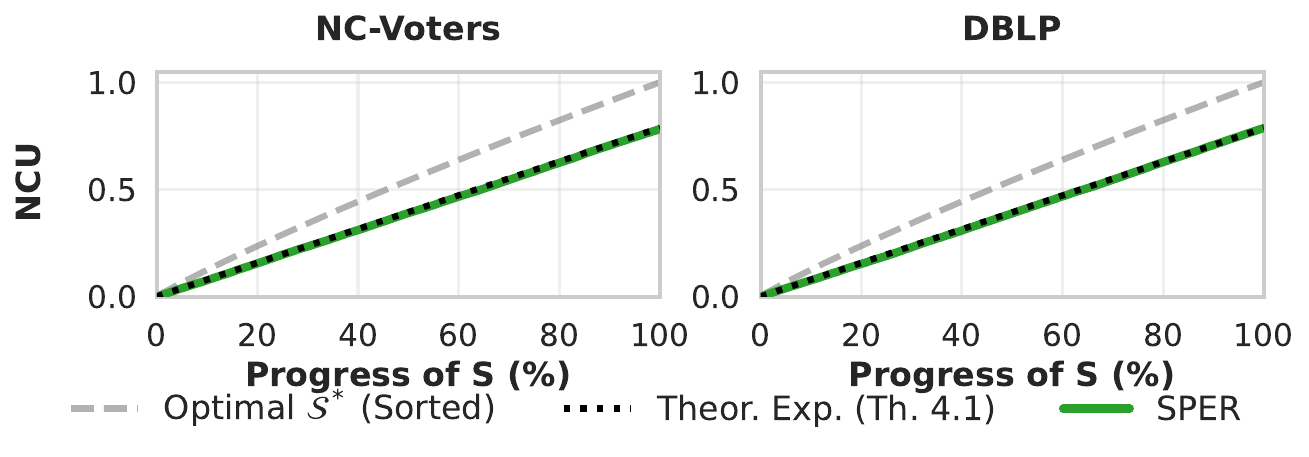}
        %\caption{Description of fourth validation}
        \label{fig:theo4}
    \end{subfigure}
    
    \caption{Comparison of SPER's cumulative utility against the theoretical expected utility model (Theorem~\ref{thm:utility}).}
    \label{fig:theorem_validation}
\end{figure}

\begin{figure*}[t]
    \centering
    \includegraphics[width=\textwidth]{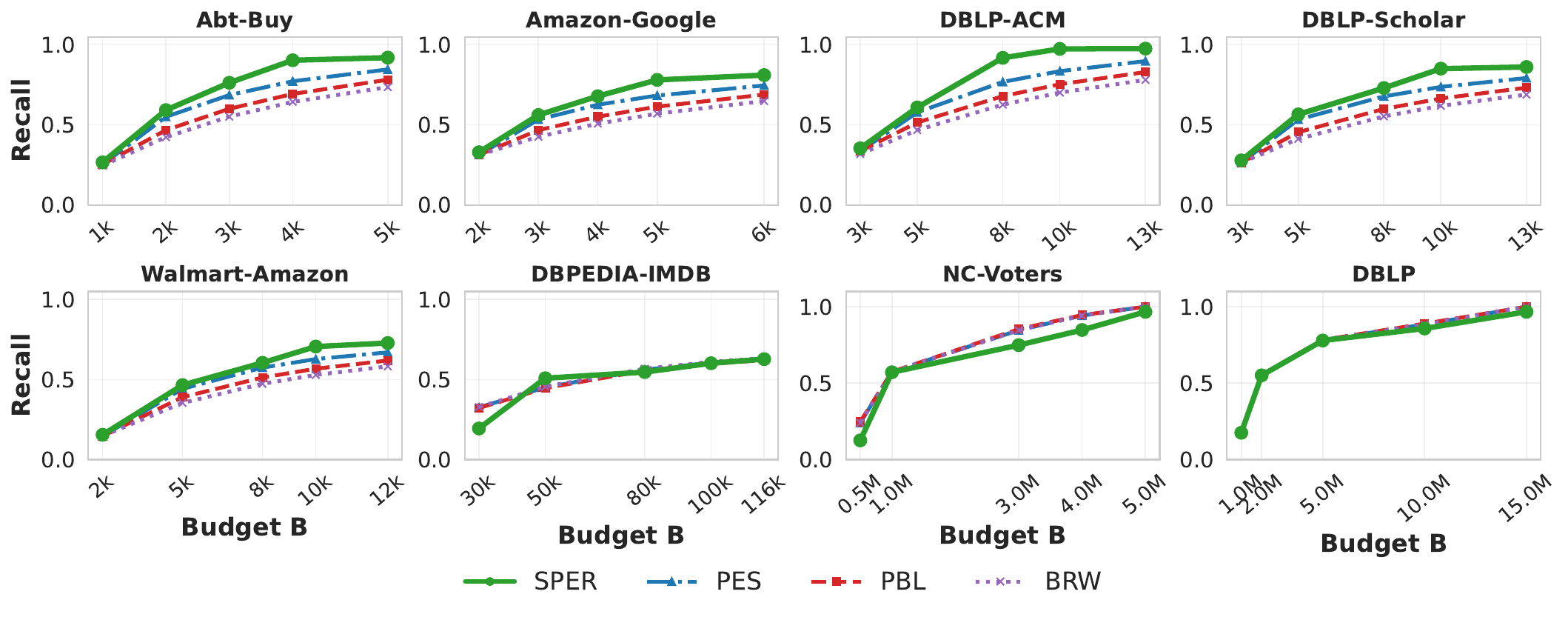}
    \caption{Cumulative Recall Analysis. The curves illustrate the recall achieved relative to the processing budget.}
    \label{fig:recall}
\end{figure*}

\begin{figure*}[t]
    \centering
    \includegraphics[width=\textwidth]{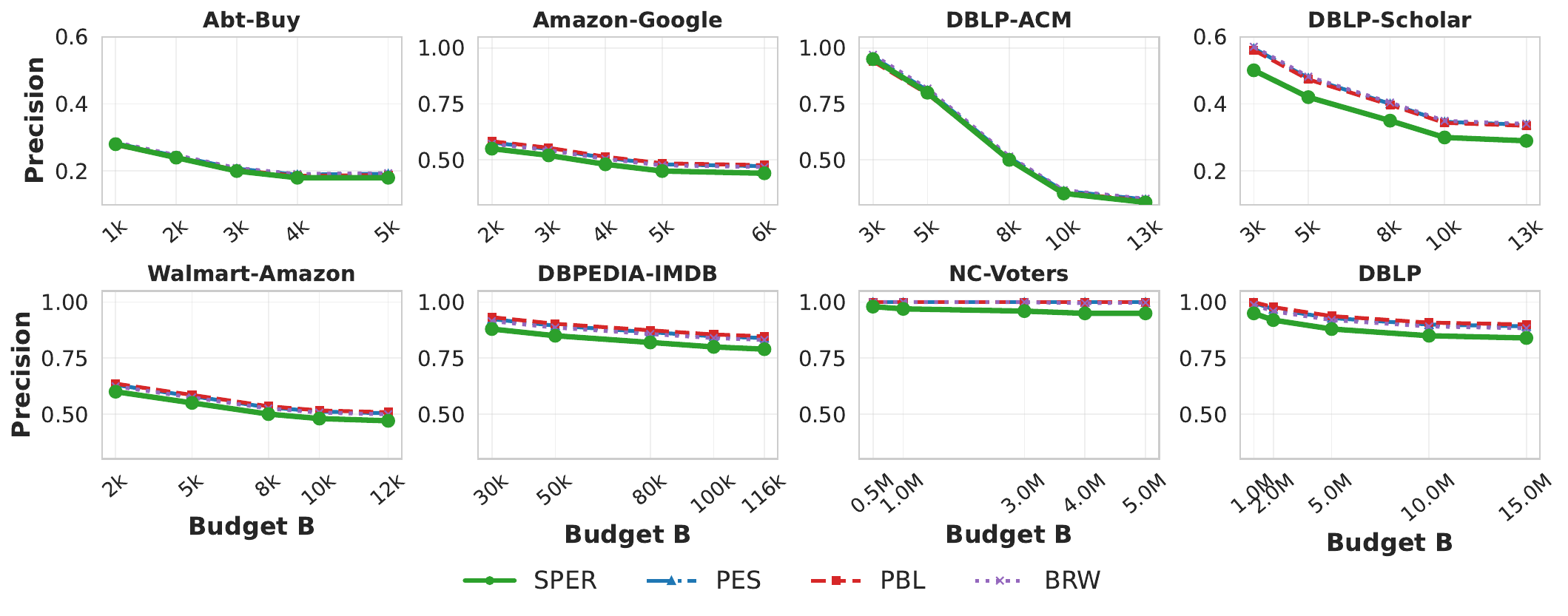}
    \caption{Precision analysis. The curves illustrate the precision achieved relative to the processing budget.} 
    \label{fig:precision}
\end{figure*}

\begin{figure*}[t]
    \centering
    \includegraphics[width=\textwidth]{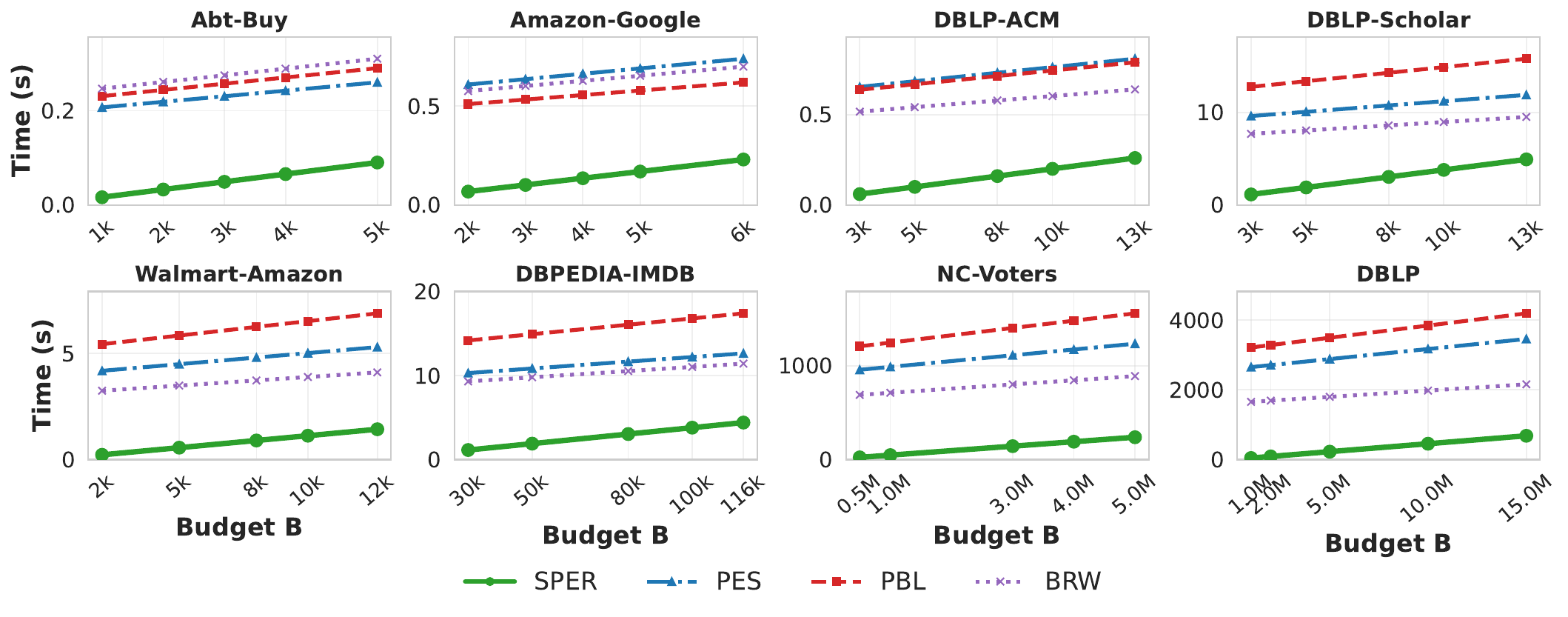}
    \caption{Comparison of the execution times (in seconds) required to prioritize the candidate space relative to the processing budget.} 
    \label{fig:time}
\end{figure*}
To bridge the gap between theoretical analysis and practical performance, we empirically validate the expected utility model established in Theorem~\ref{thm:utility} by plotting the NCU against the budget $B$. As illustrated in Figure~\ref{fig:theorem_validation}, we compare SPER against two reference baselines: the \textit{Optimal $\mathcal{S}^*$} (gray dashed line), a computationally expensive offline oracle that sorts the entire candidate set to strictly select the top-$B$ pairs, and the theoretical expectation (black dotted line), which projects the expected utility derived from the second moment of the similarity distribution ($\mathbb{E}[U] = \alpha \sum w_{(r,s)}^2$). The SPER controller (green solid line) closely tracks the theoretical trajectory in both datasets,
with slight positive deviations often observed due to the dynamic controller's ability to adapt $\alpha$ locally. This strong alignment confirms that the Stochastic Filter operates as predicted: rather than acting as a random sampler, it functions as a high-pass utility filter.

\begin{figure*}[t]
    \centering
    \includegraphics[width=\textwidth]{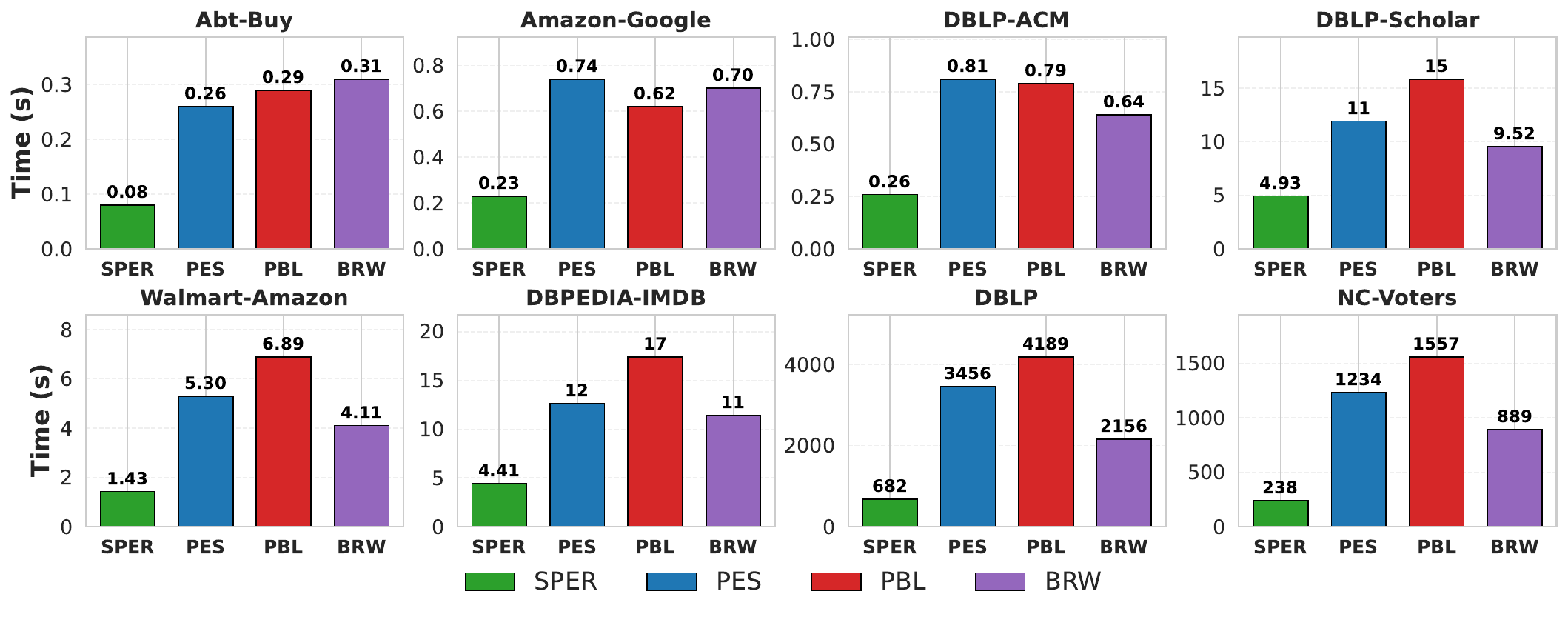}
    \caption{Comparison of the total execution times (in seconds).} 
    \label{fig:total_time}
\end{figure*}

Figure~\ref{fig:recall} demonstrates that SPER and the baselines deliver almost comparable recall for the smallest budgets across all benchmarks. We report $B$ in absolute terms, representing the direct output determined by the corresponding relative factor $\rho$. A steeper curve indicates superior \textit{progressiveness}, as it signifies that a higher percentage of true matches is identified earlier in the emission process. However, as the budget increases, SPER consistently outperforms its competitors on the more complex datasets, showing significant average improvements on Abt-Buy ($12\%$), Amazon-Google ($9\%$), DBLP-ACM ($13\%$), DBLP-Scholar ($10\%$), and Walmart-Amazon ($7\%$). These performance gains stem from SPER's use of semantic embeddings, which allow it to link entities that are lexically distinct but semantically identical (e.g., entities with \textit{PVLDB} vs. \textit{Proceedings of the VLDB Endowment}). In contrast, the three baselines rely on the token overlap, which limits their recall on these semantically heterogeneous datasets. Conversely, on NC-Voters and DBLP, where the differences between matching entities arise primarily from synthetic lexical perturbations, the baselines slightly outperform SPER by $5\%$ and $1\%$, on average, respectively. We also plot the \textit{sorted baseline using embeddings}, which prioritizes candidates strictly by their semantic similarity scores in descending order. This serves as an empirical benchmark, demonstrating that SPER's probabilistic sampling sacrifices negligible effectiveness compared to the computationally expensive deterministic optimal approach.

Despite the theoretical advantage of the baselines, which rely on deterministic sorting to strictly rank candidates, SPER achieves comparable precision levels across all benchmarks (Figure~\ref{fig:precision}). Specifically, on datasets such as Abt-Buy and DBLP-ACM, SPER yields precision scores that are statistically equivalent to the exhaustive sorting methods (e.g., $0.18$ vs. $0.17$--$0.21$ for the baselines). Even on the structurally complex DBPEDIA-IMDB and DBLP-Scholar datasets, SPER maintains a competitive approximation quality. For instance, on DBLP-Scholar, it achieves a precision of $0.29$—trailing the exhaustive baselines by an average margin of only $\approx 13\%$—while successfully retrieving the majority of high-confidence matches.

The most critical advantage of SPER is revealed in the analysis of execution time in Figure~\ref{fig:time}. Across all eight benchmarks, SPER consistently achieves the lowest latency, delivering speedups ranging from 3$\times$ to over 6$\times$. On the small datasets, SPER operates nearly instantaneously, validating its design as a zero-initialization filter. For instance, on Abt-Buy, SPER completes prioritization in just $0.08$ seconds. In contrast, the baselines require significantly longer times (from $0.26$ to $0.31$ seconds) to construct inverted indices and perform initial sorting in order to deliver results. This translates to a massive relative speedup, with SPER performing over $6\times$ faster than BRW on this dataset. Similarly, on Walmart-Amazon, SPER ($1.43$ seconds) eliminates the \textit{cold start} latency entirely, outperforming PBL ($15.89$ seconds) by approximately $5\times$. On the NC-Voters, SPER ($\sim$4 minutes) 
is roughly 5$\times$ faster than PES ($\sim$20.5 minutes)
and 6.5$\times$ faster than PBL ($\sim$26 minutes).
PES relies on a deterministic dynamic buffer for prioritization, which incurs significant index maintenance overhead. A similar trend is observed on the DBLP dataset, where SPER finishes in roughly 11 minutes, 
effectively reducing the runtime by a factor of $3$ compared to the fastest baseline (BRW at $\sim$36 minutes) and outperforming the slowest one, PBL
($\sim$69 minutes), by 6$\times$. BRW is so much slower, because it adaptively prioritizes blocks with respect to a query entity, but the comparisons within each block are executed deterministically, which renders the run-time sensitive to the block size skew. Likewise, PBL prioritizes the candidate pairs by emitting blocks of increasing size, under the assumption that earlier blocks contain more promising matches. However, progressiveness is enforced at the block level: once a block is selected, all contained record pairs are deterministically compared. As a result, PBL lacks fine-grained control over the comparison budget, which thus may incur bursty costs due to large blocks. Figure~\ref{fig:total_time} details the total execution times required by each method to process the maximum allocated budget.

While SPER shows minor budget deviations on small datasets due to granularity, this error margin becomes negligible as the dataset size $|S|$ grows (e.g., $<1\%$ overshoot on DBPEDIA-IMDB, DBLP and NC-Voters). This validates the method's concentration guarantees quantified by Inequality \ref{eq:chernoff-budget}.

\section{Conclusions and Future Work}
\label{sec:conclusions}
In this work, we presented SPER, a high-velocity framework that resolves the scalability-utility trade-off in Progressive ER. By abandoning the computationally expensive guarantee of deterministic sorting in favor of stochastic bipartite maximization, SPER successfully emulates the optimal processing order without the associated initialization overhead. Our theoretical analysis and experimental results confirm that this probabilistic approach acts as an effective high-pass filter, concentrating utility on high-confidence matches while ensuring strict adherence to computational budgets. Ultimately, SPER demonstrates that for modern, web-scale data streams, stochastic approximation is not merely a compromise but a necessary evolution to achieve real-time resolution with high fidelity.

In future work, we plan to deepen the streaming capabilities of SPER in two key directions. First, we will extend the framework to support evolving target indices, allowing the reference dataset $R$ to be updated incrementally in real-time. This would enable the system to handle truly unbounded streams, rather than querying a static index. Second, we aim to enhance the budget controller's robustness to concept drift and bursty traffic by integrating lightweight time-series forecasting. This would allow the system to preemptively adjust the window parameters and scaling factor during sudden spikes in data volume or shifts in similarity distributions. 
\balance
\bibliographystyle{ACM-Reference-Format}
\bibliography{references}
\end{document}